\documentclass[aps,pra,twocolumn,superscriptaddress,10pt,nofootinbib]{revtex4-2}

\usepackage{amsmath}
\usepackage{amssymb}
\usepackage{graphicx}
\usepackage{hyperref}
\usepackage{braket}
\usepackage{colortbl}
\usepackage{pgfplots}
\pgfplotsset{compat=1.16}
\usepgfplotslibrary{groupplots}
\usetikzlibrary{calc}

\texorpdfstring

\newtheorem{proposition}{Proposition}

\newenvironment{proof}%
  {\par\medskip\noindent\textit{Proof.}\ }%
 {\hfill$\square$\par\medskip}

\newtheorem{theorem}[proposition]{Theorem}

\newtheorem{corollary}{Corollary}
\makeatletter
\@ifundefined{proof}{%
  \newenvironment{proof}{\par\noindent\textit{Proof.}\ }%
                        {\hfill$\square$\par\medskip}}{}
\makeatother

\begin{document}

\title{Half a qubit: an algebraic fractionalization}

\author{Po-Yao Chang}
\email{pychang@phys.nthu.edu.tw}
\affiliation{Department of Physics, National Tsing Hua University, Hsinchu 30013, Taiwan}

\begin{abstract}
Fractionalizing a quantum two-level system is usually
associated with encodings based on pairs of Majorana fermions---an operational fractionalization. We show an alternative algebraic fractionalization by embedding Sz\'{e}kely's classical ``half-coin'' into a non-Hermitian Krein space. 
The coefficients of $(q+pz)^{1/2}$ define a signed sequence and a normalized,
non-Hermitian biorthogonal operator describing a biorthogonal half-qubit. 
We prove that two such objects fuse into an arbitrary pure qubit through the signed Vandermonde
convolution that  the collective
$N\ge2$ vectors are null in Krein space. $L_1$ norm of the half-qubit follows in closed form, $\lVert p\rVert_1 = 2\sqrt{q}-\sqrt{q-p}$.
Its $L_1$ norm increases monotonically with the bias and attains its supremum $\sqrt{2}$ precisely at
the unbiased point $p=q=1/2$.
Interestingly, we identify two structural results as follows. 
First, number parity and the $\eta$-metric  generate a distinguished
commuting $\mathbb Z_2\times\mathbb Z_2$ subgroup.
Second, we find the $\eta$-metric obstructs any local $\eta$-self-adjoint
partner of the parity, so a half-qubit carries a $\mathbb{Z}_2$ observable but no local
$SU(2)$. The full Pauli algebra emerges only upon fusion. We then show that the
construction survives truncation of the Fock basis: the fused qubit is exact at every
cutoff, and the Vandermonde cancellation is visible in sign-weighted photon-number statistics, 
and can be tested using existing cavity and trapped-ion state-synthesis methods. Finally, we generalize this algebraic fractionalization to a $1/n$-qubit, which can be achieved by replacing the square root with an $n$th root.
\end{abstract}

\maketitle

\section{Introduction}

The fractionalization of a quantum two-level system is a cornerstone of modern condensed
matter physics and quantum information theory. In conventional frameworks, splitting a
qubit is achieved through  operational scheme by splitting it into two Majorana fermions.
To isolate a single Majorana fermion, we need to consider 
a long-range-entangled many-body state where Majorana fermions
are at the boundaries of one-dimensional topological superconductors or
within the vortex cores of $p+ip$ superfluids~\cite{Kitaev2001, Read2000, Nayak2008,
FuKane2008, Alicea2012, Sarma2015}.
There, the fractional quantum dimension $d=\sqrt{2}$ relies entirely on multi-particle
entanglement and spatial non-locality.

In this work, we consider another fractionalization from the algebraic point of view.
This alternative approach requires the concept of negative probability from the definition of a half-coin introduced by  
Sz\'{e}kely~\cite{Szekely2005}. To incorporate with negative probability in quantum systems, 
we choose an indefinite-metric representation which is non-Hermitian.
Negative probabilities have a long history in
quantum theory, from Dirac's and Wigner's early use of signed
distributions~\cite{Dirac1942, Wigner1932} to Feynman's defence of them as bookkeeping
devices~\cite{Feynman1987}, and to their modern role as a diagnostic of nonclassicality
and a resource for quantum computation~\cite{Ferrie2011, Spekkens2008, Veitch2012,
Mari2012}. We promote the particular signed measure introduced by
Sz\'{e}kely~\cite{Szekely2005, RuzsaSzekely1988} to a non-Hermitian Krein
space~\cite{Bognar1974, Mostafazadeh2010, Bender1998, Bender2007, Brody2013}.
The properties of non-Hermitian quantum systems in a non-Hermitian Krein
space are equivalent to those of the biorthogonal formalism~\cite{Brody2013}, whose
spectral and dynamical consequences have been explored extensively in recent
years~\cite{ElGanainy2018, Ashida2020}. Within this formalism the biorthogonal basis
reorganizes the spectral data of a non-Hermitian system: it yields entanglement spectra
and entropies that diagnose non-Hermitian topology and the negative or complex central charges of the
underlying non-unitary conformal field theories~\cite{Chang2020Entanglement,
Chang2022Renyi, Chang2023Relating, Chang2026Impurity, Shimizu2025, Liu2026}. 
Here the entanglement spectrum can also be negative which naturally gives the negative probability interpretation. 
In this work, we construct an exact, local algebraic fractionalization of a qubit by using the biorthogonal formalism.

We have two main findings. The first concerns the local algebra. A half-qubit inherits two inequivalent $\mathbb{Z}_2$ gradings of its Fock
space---one by number parity, one by the sign of the quasi-probability, or it is referred to as the $\eta$-metric. 
Because both generators are diagonal they commute,
so the local symmetry is a $\mathbb{Z}_2\times\mathbb{Z}_2$ group rather than the single
$\mathbb{Z}_2$ one might expect. 
In the biorthogonal formalism,
$\langle\eta\rangle_{\rm bi}=\sum_k|p_k|$ can exceed one even though
$\eta^2=\mathbb I$. This is a formal biorthogonal expectation value, not a
Born-rule expectation for a standard projective measurement.
By contrast, we prove that no local $\eta$-self-adjoint operator anticommutes with the parity: the obstruction is of finding the partner of parity is due to the indefiniteness of the metric, so the absence of a $SU(2)$.

The second finding concerns a possible experimental realization.
The formal construction uses an infinite
Fock basis with a $k^{-3/2}$ coefficient tail. At every cutoff $K\ge1$, conditioning on
the $N\le1$ code space gives the target qubit exactly; the cutoff changes the success
probability and the approximation to the full half-qubit state. 
The signed convolution is
reconstructed by assigning the known weight $s_{n_A}s_{n_B}$ to number-resolved outcomes.
We give a four-level example and identify circuit QED and trapped-ion primitives that could
implement this conditional test.

The paper is organized as follows. Sec.~\ref{sec:coin} reviews Sz\'{e}kely's half-coin
and the necessity of negative probabilities. Sec.~\ref{sec:halfqubit} embeds it in a
Krein space, defines the biorthogonal half-qubit, and establishes the Vandermonde fusion
of two halves into one qubit. Sec.~\ref{sec:arbitrary} generalizes to arbitrary bias,
identifies the boundary of existence, and evaluates the $L_1$ norm in closed
form. Sec.~\ref{sec:halfpauli} analyses the local observable algebra: the
$\mathbb{Z}_2\times\mathbb{Z}_2$ structure, the anomalous expectation value of the $\eta$-metric,
the no-go theorem for a local $X$, and the emergence of $SU(2)$ upon fusion.
Sec.~\ref{sec:realization} treats the multiphoton realization at finite truncation. 
Sec.~\ref{sec:nthroot} generalizes the construction to $n$th roots.
Sec.~\ref{sec:majorana} discusses a possible relation to Majorana fermions, and Sec.~\ref{sec:conclusion} concludes.

\section{Half of a coin and negative probability}
\label{sec:coin}

Our approach is based on Sz\'{e}kely's ``half-coin''~\cite{Szekely2005}. In
classical probability theory, the generating function of a fair coin flip is 
$g(z)=\tfrac12+\tfrac12 z$. Finding a distribution representing ``half'' of a coin is equivalent to finding a generating function $f(z)$ whose convolution with itself yields the
fair coin, $[f(z)]^2=g(z)$.

The fractional root is $f(z)=\left(\frac{1+z}{2}\right)^{1/2}$. Expanding via the
generalized binomial theorem yields an infinite series,
\begin{equation}
f(z) = \sum_{k=0}^\infty p_k z^k, \qquad
p_k = \frac{1}{\sqrt{2}} \binom{1/2}{k}.
\label{eq:pk_fair}
\end{equation}
The coefficients $p_k$ play the role of the probability of outcome $k$. For $k=0$ and
$k=1$ they are positive. For $k\ge2$ the binomial coefficient has the strictly
alternating sign $(-1)^{k-1}$, mandating negative probabilities, e.g.\
$p_2=-1/(8\sqrt2)$ and $p_3=1/(16\sqrt2)$. Since
$\left|\binom{1/2}{k}\right| = \frac{1}{2\sqrt{\pi}}k^{-3/2}\left[1+O(k^{-1})\right]$,
the series converges \emph{absolutely}, which is what makes the quantum embedding of
Sec.~\ref{sec:halfqubit} well defined. 

Classically $p_k$ is a signed measure. When two independent half-coins are tossed, the
joint probability of the sum $N$ is the discrete convolution
$Q(N)=\sum_{j=0}^N p_j p_{N-j}$. The negative terms in the infinite tail cancel exactly,
yielding $Q(0)=Q(1)=1/2$ and $Q(N)=0$ for all $N\ge2$: the negative probabilities
annihilate every higher outcome and restore a standard binary distribution
as shown in Fig.~\ref{fig:pk}.

\begin{figure}[t]
\centering
\begin{tikzpicture}[
  every axis/.style={
    width=\columnwidth, height=3.9cm,
    xmin=-0.6, xmax=8.6, xtick={0,1,2,3,4,5,6,7,8},
    tick label style={font=\footnotesize}, label style={font=\footnotesize},
    axis lines=left, ybar, bar width=7pt,
    extra y ticks={0}, extra y tick labels={}, extra y tick style={grid=major},
  }]
\begin{axis}[name=axA, ylabel={$p_k$}, xlabel={$k$}, ymin=-0.2, ymax=0.8]
\addplot[fill=blue!55,draw=blue!70!black] coordinates
 {(0,0.70711) (1,0.35355) (2,-0.08839) (3,0.04419) (4,-0.02762)
  (5,0.01933) (6,-0.01450) (7,0.01139) (8,-0.00926)};
\end{axis}
\begin{axis}[
  at={($(axA.north east)+(-3.1cm,0.6cm)$)}, anchor=north west,
  width=2.9cm, height=1.9cm, scale only axis,
  xmin=1.4, xmax=8.6, ymin=-0.115, ymax=0.075,
  xtick={2,4,6,8}, ytick={-0.08,0,0.04},
  yticklabel style={/pgf/number format/fixed, /pgf/number format/precision=2},
  tick label style={font=\tiny}, bar width=2.6pt,
  title={\tiny $k\ge2$, magnified}, title style={yshift=-5pt},
]
\addplot[fill=blue!55,draw=blue!70!black] coordinates
 {(2,-0.08839) (3,0.04419) (4,-0.02762) (5,0.01933)
  (6,-0.01450) (7,0.01139) (8,-0.00926)};
\end{axis}

\begin{axis}[
  at={($(axA.south west)-(0,1.15cm)$)}, anchor=north west,
  ylabel={$Q(N)=\sum_j p_jp_{N-j}$}, xlabel={$N$}, ymin=-0.1, ymax=0.65]
\addplot[fill=red!55,draw=red!70!black] coordinates
 {(0,0.5) (1,0.5) (2,0) (3,0) (4,0) (5,0) (6,0) (7,0) (8,0)};
\end{axis}
\end{tikzpicture}
\caption{Sz\'{e}kely quasi-probabilities of the fair half-coin, Eq.~\eqref{eq:pk_fair} 
(top; inset magnifies the alternating tail $k\ge2$), and their self-convolution (bottom).
The negative tail leads the collective $N\ge2$ vectors to be Krein-null exactly, leaving
$Q(N)=\tfrac12\binom{1}{N}$.}
\label{fig:pk}
\end{figure}
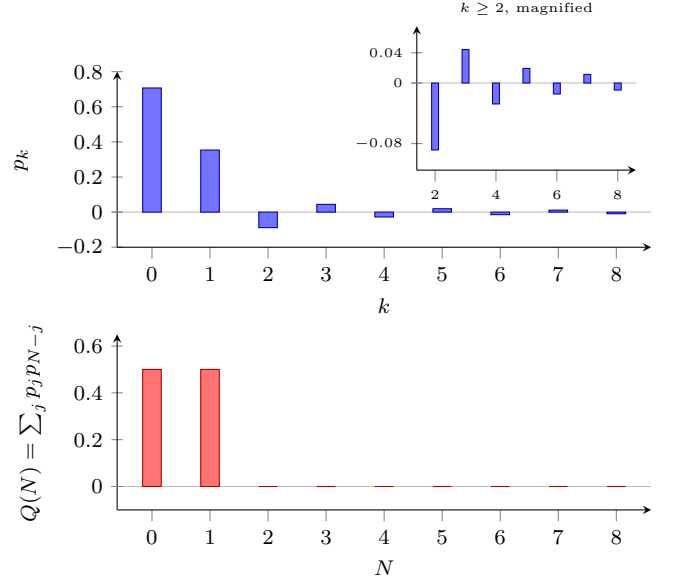

\section{The biorthogonal half-qubit}
\label{sec:halfqubit}

\subsection{Krein embedding}

To fractionalize a quantum state we elevate the classical signed probabilities to
probability amplitudes in a Fock space $\{\ket{k}\}_{k=0}^\infty$. Because the $p_k$ may
be negative, the state cannot be accommodated in a Hilbert space with a positive-definite
inner product; it resides instead in a Krein space~\cite{Bognar1974, Mostafazadeh2010},
an indefinite-metric space in which vectors may have negative norm. We introduce the
$\eta$-metric operator
\begin{equation}
\eta = \sum_{k=0}^\infty s_k \ket{k}\bra{k}, \qquad s_k \equiv \operatorname{sgn}(p_k),
\label{eq:eta}
\end{equation}
so that $\eta^\dagger=\eta$ and $\eta^2=\mathbb{I}$.
Here the $\eta$-metric is also used in pseudo-Hermitian quantum system~\cite{Mostafazadeh2010}.
We also emphasize that this $\eta$-metric operator is for any form of half-qubit as we will discuss in Sec.~\ref{sec:arbitrary}.
 Explicitly the $\operatorname{sgn}(p_k)$ are 
 \begin{equation}
s_0=+1, \quad s_k=+1 \ (k \text{ odd}), \quad s_k=-1 \ (k \text{ even},\, k\ge2),
\label{eq:signs}
\end{equation}
so the metric is a fixed, state-independent operator with infinitely many positive and
infinitely many negative directions.

With this Krein structure in place we may use the biorthogonal formulation of quantum
mechanics~\cite{Brody2013, Mostafazadeh2010} to define the single-mode right and left
half-qubit states as
\begin{align}
\ket{R} &= \sum_{k=0}^\infty \sqrt{|p_k|}\, e^{ik\phi} \ket{k}, \label{eq:R}\\
\bra{L} &= \sum_{k=0}^\infty s_k \sqrt{|p_k|}\, e^{-ik\phi} \bra{k}, \label{eq:L}
\end{align}
where $\phi$ is the phase carrying the quantum coherence. The two states are
not independent but are related by the $\eta$-metric,
\begin{equation}
\ket{L} =\eta \ket{R},
\qquad\text{equivalently}\qquad
\bra{L} = \ket{R}^\dagger\eta ,
\label{eq:LisRdaggereta}
\end{equation}
so that Eqs.~\eqref{eq:R}--\eqref{eq:L} constitute a standard $\eta$-pseudo-Hermitian
pair rather than two independent objects. The biorthogonal density operator is
$\rho_{\text{bi}}=\ket{R}\bra{L}$, and the trace is automatically
normalized,
\begin{equation}
\operatorname{Tr}\rho_{\text{bi}} = \braket{L|R} = \sum_{k=0}^\infty p_k = f(1) = 1 .
\end{equation}
Note that $\ket{R}$ is \emph{not} normalized with respect to the ordinary Hilbert-space
inner product: $\braket{R|R}=\sum_k|p_k| \equiv \lVert p\rVert_1 \geqslant 1$, a fact we exploit in Sec.~\ref{sec:arbitrary} and must account for in Sec.~\ref{sec:realization}.
Here $\lVert p\rVert_1$ is the $L_1$ norm which is the sum  of the absolute values $\sum_k|p_k|$ over all Fock levels.

\subsection{Fusion of two half-qubits}
\label{sec:composite_qubit}

We now form the bipartite tensor product of two half-qubits $A$ and $B$, prepared with a
phase $\phi$. The global
density matrix is
\begin{equation}
\rho_{AB} = \rho_{\text{bi}}^{(A)} \otimes \rho_{\text{bi}}^{(B)}
          =  \ket{R_{AB}}\bra{L_{AB}},
\end{equation}
with
\begin{align}
\ket{R_{AB}} &= \sum_{N=0}^\infty \sum_{k=0}^N
   \sqrt{|p_k||p_{N-k}|}\, e^{iN\phi} \ket{k,N-k}, \\
\bra{L_{AB}} &= \sum_{M=0}^\infty \sum_{l=0}^M
   s_l s_{M-l} \sqrt{|p_l||p_{M-l}|}\, e^{-iM\phi} \bra{l,M-l}.
\end{align}
Here $N$ denotes the total number of the composite. The physically accessible
observable when the two halves are brought together is not the microscopic pair
$(k,N-k)$ but the collective manifold spanned by all partitions of $N$.

Decomposing the global state into these macroscopic sectors,
\begin{equation}
\rho_{AB} = \sum_{N,M} e^{i(N-M)\phi}\, \ket{\bar N}_R\, {}_L\!\bra{\bar M},
\label{eq:rhoABsectors}
\end{equation}
where
\begin{align}
\ket{\bar N}_R &= \sum_{k=0}^N \sqrt{|p_k||p_{N-k}|}\,\ket{k,N-k}, \\
{}_L\!\bra{\bar M} &= \sum_{k=0}^M s_k s_{M-k}\sqrt{|p_k||p_{M-k}|}\,\bra{k,M-k},
\end{align}
are the unnormalized collective right and left states. Their biorthogonal overlap is
exactly diagonal and returns the population weight of each manifold,
\begin{equation}
{}_L\!\braket{\bar N | \bar M}_R = \delta_{NM}\, Q(N)
= \delta_{NM} \sum_{j=0}^N p_j p_{N-j}.
\label{eq:overlap}
\end{equation}
Applying Vandermonde's convolution identity to Eq.~\eqref{eq:overlap} gives
$Q(N)=\tfrac12\binom{1}{N}$, so that $Q(0)=Q(1)=1/2$ while, crucially, $Q(N)=0$ for all
$N\ge2$.

The vanishing of the higher manifolds forces the introduction of a normalized logical
basis. The operator $\rho_{AB}$ lives in an infinite-dimensional Krein space whose
macroscopic sectors are unnormalized, with
${}_L\!\braket{\bar 0|\bar 0}_R={}_L\!\braket{\bar 1|\bar 1}_R=1/2$ and
${}_L\!\braket{\bar N|\bar N}_R=0$ for $N\ge2$. Any physically observable subsystem must
however obey the axioms of standard quantum mechanics. We therefore define, \emph{for the
non-null manifolds only},
\begin{align}
\label{eq:log}
&\ket{N_{\text{log}}}_R = \frac{\ket{\bar N}_R}{\sqrt{Q(N)}}, \qquad
{}_L\!\bra{N_{\text{log}}} = \frac{{}_L\!\bra{\bar N}}{\sqrt{Q(N)}}, \\ \notag
&N \in \{0,1\},
\end{align}
which by Eq.~\eqref{eq:overlap} satisfy
${}_L\!\braket{N_{\text{log}}|M_{\text{log}}}_R=\delta_{NM}$. 
This property is the {\it collective radical-quotient theorem}~[see App. \ref{APP}].
The radical reduction of the collective subspace produces a two-dimensional Hilbert space canonically isomorphic, after choosing the normalized quotient basis, to the state space of a qubit.

The associated
\begin{equation}
\mathcal{P} = \sum_{N=0}^{1} \ket{N_{\text{log}}}_R\, {}_L\!\bra{N_{\text{log}}},
\label{eq:proj}
\end{equation}
is a $\eta$-self-adjoint projector. 

The effective physical density matrix is obtained by projecting the global state onto
this observable basis, $\rho_{\text{eff}}=\mathcal{P}\rho_{AB}\mathcal{P}$, with matrix
elements $\rho_{NM}={}_L\!\bra{N_{\text{log}}}\rho_{AB}\ket{M_{\text{log}}}_R$.
Substituting Eq.~\eqref{eq:rhoABsectors} and using
${}_L\!\braket{N_{\text{log}}|\bar K}_R=\delta_{NK}\sqrt{Q(N)}$ yields directly
\begin{equation}
\rho_{NM} = \sqrt{Q(N)Q(M)}\; e^{i(N-M)\phi}.
\label{eq:element}
\end{equation}
We stress that no renormalization is performed by hand: because $Q(0)+Q(1)=1$, the
projection is automatically trace-preserving,
$\operatorname{Tr}(\mathcal{P}\rho_{AB}\mathcal{P})=1$.

Since the Vandermonde convolution makes each collective $N\ge2$ vector null in the restricted Krein form,
the observable state space is
truncated to $N,M\in\{0,1\}$ and, in the logical basis,
\begin{equation}
\rho_{\text{eff}} =
\begin{pmatrix} 1/2 & \tfrac12 e^{-i\phi} \\[4pt] \tfrac12 e^{i\phi} & 1/2 \end{pmatrix},
\end{equation}
a pure, fully coherent, positive-definite $2\times2$ density matrix with unit trace and
vanishing determinant. Hence we compose two half-qubits with proper normalization from a standard Hermitian qubit.

\paragraph*{Remark (phase locking).}
The sector decomposition above requires the two halves to share the phase $\phi$. If they
carry independent phases $\phi_A\ne\phi_B$, the $N=1$ component of
$\ket{R_A}\otimes\ket{R_B}$ becomes
$\sqrt{|p_0p_1|}\left(e^{i\phi_B}\ket{0,1}+e^{i\phi_A}\ket{1,0}\right)$,
which is no longer proportional to $\ket{\bar 1}_R$. A purely local relative dephasing
therefore removes the composite from the logical code space.

\section{Arbitrary pure qubits, the existence boundary, and the $L_1$ norm}
\label{sec:arbitrary}

\subsection{Generalized quasi-probabilities}

We generalize to an arbitrary pure qubit with ground- and excited-state weights $q$ and $p$,
$p+q=1$. The classical generating function is $g(z)=q+pz$, its fractional root is
$f(z)=(q+pz)^{1/2}$, and the generalized Sz\'{e}kely quasi-probabilities are
\begin{equation}
p_k = \sqrt{q}\,\binom{1/2}{k}\left(\frac{p}{q}\right)^{k}.
\label{eq:pk_gen}
\end{equation}

\paragraph*{Domain of existence.}
Equation~\eqref{eq:pk_gen} defines a signed measure of bounded total variation if and only
if $\sum_k \left|\binom{1/2}{k}\right| (p/q)^k$ converges, i.e.\ if and only if $p \le q$.
For $p>q$ the radius of convergence $q/p$ of the binomial series is smaller than unity,
the coefficients grow without bound, and no half-qubit exists. This restriction is not a
loss of generality---relabelling $\ket{0}\leftrightarrow\ket{1}$ maps $p\leftrightarrow q$
and brings any qubit into the admissible range.

\subsection{Vandermonde fusion for arbitrary bias}

Composing two such half-qubits, the weight of the $N$-th macroscopic manifold is
\begin{equation}
Q(N) = \sum_{j=0}^N p_j p_{N-j}
     = q\left(\frac{p}{q}\right)^{N}
       \sum_{j=0}^{N}\binom{1/2}{j}\binom{1/2}{N-j},
       \label{eq:QNdef}
\end{equation}
where the prefactors independent of $j$ have been extracted. Vandermonde's identity
$\sum_{j=0}^N \binom{\alpha}{j}\binom{\beta}{N-j}=\binom{\alpha+\beta}{N}$ evaluates the
sum to $\binom{1}{N}$, so
\begin{equation}
Q(N) = q\left(\frac{p}{q}\right)^{N}\binom{1}{N}.
\end{equation}
Hence $Q(0)=q$, $Q(1)=p$, and $Q(N)=0$ for every $N\ge2$ because $\binom{1}{N}$ vanishes
identically. The matrix elements in the logical basis of Eq.~\eqref{eq:log} again take the
form of Eq.~\eqref{eq:element}, giving
\begin{equation}
\rho_{\text{eff}} =
\begin{pmatrix} q & \sqrt{pq}\, e^{-i\phi} \\[4pt] \sqrt{pq}\, e^{i\phi} & p \end{pmatrix},
\end{equation}
as the density matrix of the pure state $\sqrt{q}\ket{0}+\sqrt{p}\,e^{i\phi}\ket{1}$. Two infinite-dimensional
fractional objects therefore recombine into an arbitrary pure qubit.

\subsection{The $L_1$ norm}

To quantify the non-Hermiticity of the fractionalization we evaluate the $L_1$ norm of
the quasi-probability distribution, which coincides with the ordinary Hilbert-space norm
of the right state,
\begin{equation}
\braket{R|R} = \lVert p\rVert_1 = \sum_{k=0}^\infty |p_k| .
\end{equation}
This is the quantity whose logarithm defines the quantum resource of a quasi-probability representation of stabilizer computation~\cite{Veitch2014,
Veitch2012, Mari2012}.

Using the sign pattern~\eqref{eq:signs}, $|\binom{1/2}{k}|=(-1)^{k-1}\binom{1/2}{k}$ for
$k\ge1$, and separating the $k=0$ term,
\begin{equation}
\lVert p\rVert_1
= \sqrt{q}\left(1 - \sum_{k=1}^\infty \binom{1/2}{k}\left(-\frac{p}{q}\right)^{k}\right).
\end{equation}
The remaining sum is the Taylor expansion of $\sqrt{1-p/q}$ minus its constant term, and
the exact closed form follows:
\begin{equation}
\lVert p\rVert_1 = 2\sqrt{q}-\sqrt{q-p},
\label{eq:L1}
\end{equation}
valid when $p \le q$.

Two features of Eq.~\eqref{eq:L1} deserve emphasis  as shown in Fig.~\ref{fig:L1}. 
First, for a deterministic vacuum ($p=0$, $q=1$) the norm is exactly $1$: no negative probabilities are
required, and the ``half-qubit'' degenerates to a single level system. Second, writing
$q=1-p$,
\begin{equation}
\frac{d\lVert p\rVert_1}{dp} = \frac{1}{\sqrt{1-2p}}-\frac{1}{\sqrt{1-p}} > 0
\qquad \text{for } 0<p<\tfrac12,
\end{equation}
so the norm increases monotonically and approaches its
\emph{supremum}
\begin{equation}
\sup_{p} \lVert p\rVert_1 = \lVert p \rVert_1\big|_{p=q=1/2} = \sqrt{2}
\end{equation}
exactly at the endpoint $p=1/2$. 

\begin{figure}[t]
\centering
\begin{tikzpicture}
\begin{axis}[
  width=\columnwidth, height=5.2cm,
  xlabel={$p$}, ylabel={$\lVert p\rVert_1$},
  xmin=0, xmax=0.5, ymin=0.95, ymax=1.5,
  xtick={0,0.1,0.2,0.3,0.4,0.5},
  tick label style={font=\footnotesize}, label style={font=\footnotesize},
  axis lines=left, clip=false,
  ytick={1.0,1.1,1.2,1.3,1.4142}, yticklabels={$1.0$,$1.1$,$1.2$,$1.3$,$\sqrt{2}$},
]
\addplot[blue!70!black, very thick, domain=0:0.49999, samples=300]
  {2*sqrt(1-x) - sqrt(1-2*x)};
\addplot[red, only marks, mark=*, mark size=2pt] coordinates {(0.5,1.41421)};
\addplot[gray, dashed, forget plot] coordinates {(0,1.41421) (0.5,1.41421)};
\end{axis}
\end{tikzpicture}
\caption{$L_1$ norm of the half-qubit, Eq.~\eqref{eq:L1}, plotted over the
domain of existence $0\le p\le q$. The norm rises monotonically and reaches $\sqrt{2}$
exactly at the unbiased point $p=q=1/2$ (red marker), which is simultaneously the boundary
beyond which the binomial series diverges and no half-qubit exists.}
\label{fig:L1}
\end{figure}
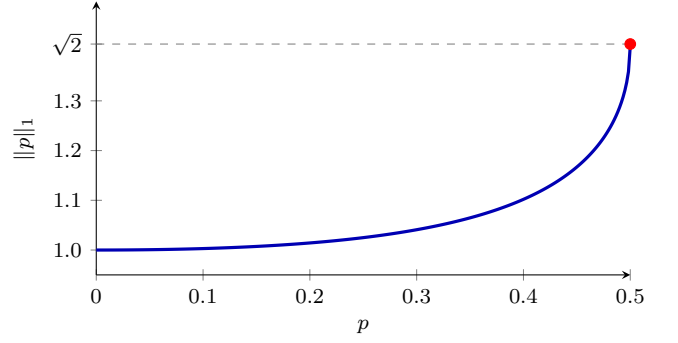

\section{Local observables: the half-Pauli algebra}
\label{sec:halfpauli}

Because a half-qubit occupies an infinite-dimensional Fock space, it does not carry a
closed $SU(2)$ algebra by construction. It is natural to ask which bounded
operators on the Krein space can serve as local observables. Recall that an operator $O$
is an observable in a Krein space when it is $\eta$-self-adjoint,
\begin{equation}
O^\sharp \equiv \eta\, O^\dagger \eta = O ,
\label{eq:etaselfadj}
\end{equation}
which guarantees that its biorthogonal expectation value
$\braket{O}=\bra{L}O\ket{R}=\bra{R}\eta O\ket{R}$ is real.

\subsection{A distinguished
\texorpdfstring{$\mathbb{Z}_2\times\mathbb{Z}_2$}{Z2 x Z2} grading subgroup}

We take the first grading to be number parity,
\begin{equation}
Z_A = (-1)^{\hat n_A} = \sum_{k=0}^\infty (-1)^k \ket{k}\bra{k},
\label{eq:z2}
\end{equation}
which is diagonal, commutes with $\eta$, satisfies $Z_A^\sharp=Z_A$, and obeys
$Z_A^2=\mathbb I$. Its biorthogonal expectation
value is the generating function evaluated at $z=-1$,
\begin{equation}
\langle Z_A\rangle_{\rm bi} = \sum_{k=0}^\infty (-1)^k p_k
= f(-1) = \sqrt{q-p},
\label{eq:ZA}
\end{equation}
which is real throughout the domain $p\le q$. Composing two halves gives
$Z_A\otimes Z_B=(-1)^{\hat n_A+\hat n_B}=(-1)^N$ and
\begin{equation}
\langle Z_A\otimes Z_B\rangle_{\rm bi} = \sum_N (-1)^N Q(N) = q-p,
\end{equation}
the usual $\sigma_z$ expectation value of the reconstructed qubit.

Comparing the parity of Eq.~\eqref{eq:z2} with the $\eta$-metric in the form~\eqref{eq:eta} gives
\begin{equation}
  \eta=2|0\rangle\langle0|-Z_A ,
  \label{eq:etaZA}
\end{equation}
so the two differ by a rank-one operator, i.e., the vacuum state. 
Each defines a $\mathbb{Z}_2$ grading, and the
gradings differ \emph{only} in the placement of $|0\rangle$,
\begin{align}
  Z_A:&\quad \{0,2,4,6,\dots\}\;\oplus\;\{1,3,5,\dots\},
  \label{eq:gradingZ}\\
  \eta:&\quad \{0,1,3,5,\dots\}\;\oplus\;\{2,4,6,\dots\}.
  \label{eq:gradingeta}
\end{align}

Since $Z_A$ and $\eta$ are diagonal, they commute and are both
$\sharp$-self-adjoint. Their product is
\begin{equation}
  V\equiv Z_A\,\eta=2|0\rangle\langle0|-I ,
\end{equation}
The four involutions therefore form the Klein group
$\mathfrak G=\{I,Z_A,\eta,V\}\cong\mathbb Z_2\times\mathbb Z_2$. 
Their
biorthogonal expectation values are
\begin{equation}
  \langle Z_A\rangle_{\rm bi}=\sqrt{q-p},\qquad
  \langle V\rangle_{\rm bi}=2\sqrt q-1,\qquad
  \langle \eta\rangle_{\rm bi}=\lVert\boldsymbol p\rVert_1 .
  \label{eq:threeZ2}
\end{equation}
Although $\eta$ has ordinary spectrum $\{\pm1\}$, its formal biorthogonal value is
$\langle\eta\rangle_{\rm bi}=\bra R\eta^2\ket R=\braket{R|R}
=\lVert\boldsymbol p\rVert_1$, which can exceed one. This is not a contradiction:
$\operatorname{Tr}(\rho_{\rm bi}\,\cdot)$ is not a positive state on the ordinary
Hilbert-space operator algebra. We use this value as an algebraic Krein-space
diagnostic, not as the Born expectation of a conventional dichotomic measurement.
Equation~\eqref{eq:etaZA} also gives
$\lVert\boldsymbol p\rVert_1=2p_0-\langle Z_A\rangle_{\rm bi}
=2\sqrt q-\sqrt{q-p}$ directly.

\subsection{No local partner of the parity exists}
\label{sec:nolocalX}

One would like a local $X_A$ with $X_A^2=\mathbb{I}$ and $\{X_A,Z_A\}=0$. The obvious
candidate is the operator that swaps adjacent levels,
\begin{equation}
X_A^{} = \sum_{k=0}^\infty
\Big( \ket{2k+1}\bra{2k} + \ket{2k}\bra{2k+1} \Big),
\label{eq:X0}
\end{equation}
which indeed satisfies both requirements. It is not, however, an observable: since
$s_{2k}s_{2k+1}=-1$ for every $k\ge1$ while $s_0s_1=+1$, one has
$\eta X_A^{}\eta \ne X_A^{}$, and the biorthogonal expectation value is complex,
\begin{equation}
\braket{X_A^{}} = 2\sqrt{p_0p_1}\,\cos\phi
 \;-\; 2i\sin\phi \sum_{k=1}^\infty \sqrt{|p_{2k}p_{2k+1}|}.
\label{eq:XA}
\end{equation}
For the unbiased half-qubit, $p_0p_1=1/4$ so the real part is exactly $\cos\phi$, while the
sum in the imaginary part converges---slowly, as $k^{-1/2}$ in the partial sums---to
\begin{equation}
2\sum_{k=1}^\infty \sqrt{|p_{2k}p_{2k+1}|} = 0.34489\ldots ,
\end{equation}
giving $\braket{X_A^{}}=\cos\phi - 0.34489\, i\sin\phi$.

The complex expectation value in Eq.~\eqref{eq:XA} shows that the
candidate operator in Eq.~\eqref{eq:X0} is not
$\eta$-self-adjoint and therefore cannot serve as a Krein-space
observable. This failure is not specific to that particular ansatz.
Within the present construction, no bounded $\eta$-self-adjoint
involution on the full Fock space of a single half-qubit can
anticommute with the number-parity operator $Z_A$. Thus, a global
half-Pauli $X_A$ satisfying the usual Pauli relations does not exist
under these assumptions, as established by the following proposition.

\begin{proposition}[Parity--metric obstruction]
\label{prop:nogo}
Let $\mathcal H$ be the Fock space of a single
half-qubit. Let $Z_A$ be the number-parity operator and let $\eta$-metric defined by the sign pattern in
Eq.~\eqref{eq:signs}. There exists no bounded operator
$X_A$ satisfying simultaneously
\begin{align}
\{X_A,Z_A\}&=0, \label{eq:nogo-anticommute}\\
X_A^2&=\mathbb I, \label{eq:nogo-involution}\\
X_A^\sharp\equiv\eta X_A^\dagger\eta&=X_A.
\label{eq:nogo-selfadjoint}
\end{align}
Consequently, $Z_A$ has no bounded $\eta$-self-adjoint Pauli-$X$
partner acting on the full single-half-qubit Fock space.
\end{proposition}

 \begin{proof}
Assume that such an operator $X_A$ exists. Define the biorthogonal inner
product by
\[
[u,v]_\eta=\langle u|\eta v\rangle .
\]
Because $X_A^\sharp=X_A$ and $X_A^2=\mathbb I$, we have
\[
X_A^\dagger\eta X_A
 =\eta X_A^2
 =\eta .
\]
Hence $X_A$ preserves the indefinite inner product:
\[
[X_Au,X_Av]_\eta=[u,v]_\eta .
\]

Now consider the even-parity state $\ket{2}$. From the sign pattern of
$\eta$,
\[
Z_A\ket{2}=\ket{2},
\qquad
[\ket{2},\ket{2}]_\eta=-1.
\]
The anticommutation relation implies
\[
Z_A X_A\ket{2}=-X_A\ket{2},
\]
so $X_A\ket{2}$ belongs to the odd-parity subspace. The metric $\eta$
is positive on the entire odd-parity subspace. Moreover,
$X_A\ket{2}\neq0$ because $X_A^2=\mathbb I$. Therefore,
\[
[X_A\ket{2},X_A\ket{2}]_\eta
 =\lVert X_A\ket{2}\rVert^2>0.
\]
This contradicts preservation of the indefinite inner product, which
would require
\[
[X_A\ket{2},X_A\ket{2}]_\eta
 =[\ket{2},\ket{2}]_\eta=-1.
\]
Thus no such operator exists.
\end{proof}

The proof identifies the source of the obstruction as the mismatch
between the signatures of $\eta$ on the even- and odd-parity
subspaces. The metric is positive on every odd level, whereas the
even-parity subspace contains one positive direction, $\ket{0}$, and
infinitely many negative directions, $\ket{2},\ket{4},\ldots$.
An $\eta$-self-adjoint involution is an $\eta$-isometry, while
anticommutation with $Z_A$ would require it to map the even- and
odd-parity subspaces bijectively onto one another. Their inequivalent
metric signatures make such an isometric mapping impossible. If
$\eta$ were positive definite, this particular obstruction would
disappear; however, the negative coefficients could then no longer be
represented as signed norms in the present construction. We therefore
obtain the following structural conclusion:

\begin{quote}
\emph{The Pauli algebra is not realized globally on an isolated half-qubit
with $Z_A$ chosen as number parity; within the present construction,
it is recovered on the fused two-dimensional code space.}
\end{quote}

\subsection{Encoded Pauli operators on the fused qubit}

Although $X_A$ is not a local Krein observable, 
 it can be used as an ambient operator
whose compression acts on the selected fused qubit.
a macroscopic $X_L$ does, once the composite is
projected onto the physical subspace. Define
\begin{equation}
\Sigma_x = \tfrac{1}{\sqrt2}\left(X_A^{}+X_B^{}\right).
\end{equation}
The logical states of Eq.~\eqref{eq:log} are, for \emph{every} admissible bias,
\begin{equation}
\ket{0_{\text{log}}}_R = \ket{0,0}, \qquad
\ket{1_{\text{log}}}_R = \tfrac{1}{\sqrt2}\left(\ket{0,1}+\ket{1,0}\right).
\label{eq:logexplicit}
\end{equation}
This is not special to the unbiased point. From Eq.~\eqref{eq:pk_gen},
$p_0=\sqrt q$ and $p_1=p/(2\sqrt q)$, so $\ket{\bar0}_R=\sqrt q\,\ket{0,0}$ and
$\ket{\bar1}_R=\sqrt{p/2}\,(\ket{0,1}+\ket{1,0})$, while $Q(0)=q$ and $Q(1)=p$; the
normalization $\sqrt{Q(N)}$ in Eq.~\eqref{eq:log} therefore cancels the bias dependence
exactly, and the same holds for the left states because $s_0=s_1=+1$. Consequently
$\mathcal{P}$ is the ordinary orthogonal projector onto
$\operatorname{span}\{\ket{0,0},\tfrac{1}{\sqrt2}(\ket{0,1}+\ket{1,0})\}$, independent of
$(p,q)$, and the computation below establishes Eq.~\eqref{eq:Sigmax} for the whole family
of half-qubits at once.
A direct computation gives
\begin{align}
(X_A+X_B)\ket{0_{\text{log}}}_R &= \ket{1,0}+\ket{0,1}
   = \sqrt2\,\ket{1_{\text{log}}}_R, \\
(X_A+X_B)\ket{1_{\text{log}}}_R &= \sqrt2\,\ket{0,0} + \sqrt2\,\ket{1,1}.
\end{align}
The first line lies entirely in the physical subspace. In the second, the component
$\ket{1,1}$ belongs to the $N=2$ manifold and is removed by $\mathcal{P}$ of
Eq.~\eqref{eq:proj}, since ${}_L\!\braket{N_{\text{log}}|1,1}=0$ for $N=0,1$. Hence
$\mathcal{P}(X_A+X_B)\ket{1_{\text{log}}}_R=\sqrt2\ket{0_{\text{log}}}_R$,
and with the fusion normalization $1/\sqrt2$,
\begin{equation}
 X_{\rm L}\equiv \mathcal{P}\,\Sigma_x\,\mathcal{P} = \sigma_x ,
\label{eq:Sigmax}
\end{equation}

Likewise,
\begin{align}
 Z_{\rm L}&\equiv\left.\mathcal P(Z_A\otimes Z_B)\mathcal P
 \right|_{\operatorname{Ran}\mathcal P}=\sigma_z,\\
 Y_{\rm L}&\equiv-iZ_{\rm L}X_{\rm L}=\sigma_y.
\end{align}
Thus the selected positive code supports the usual Pauli algebra. In the fused qubit state,
\begin{equation}
 \langle X_{\rm L}\rangle=2\sqrt{pq}\cos\phi,\quad
 \langle Y_{\rm L}\rangle=2\sqrt{pq}\sin\phi,\quad
 \langle Z_{\rm L}\rangle=q-p.
\end{equation}
In particular, $\langle X_{\rm L}\rangle=\cos\phi$ only at $p=q=1/2$.

\section{MULTIPHOTON REALIZATION AT FINITE TRUNCATION}
\label{sec:realization}

In this section we discuss physical realization of the half-qubit, which relies on
 showing that the truncation is far less
damaging than it first appears. The fusion of two half-qubits into a
standard qubit is \emph{exact at every cutoff}, and only quantities
built from the $L_1$ norm inherit the slow $ k^{-3/2} $ tail.

\subsection{Physical setting}

We take as reference platform a pair of microwave cavity modes
dispersively coupled to transmon ancillas, although the discussion
applies verbatim to the motional modes of a trapped ion and, with the
Fock index replaced by a frequency-bin index, to synthetic-dimension
photonics.

Three primitives are required.

\emph{(i) Preparation.} The truncated right state
\begin{equation}
  |R_K\rangle=\mathcal{N}_K^{-1/2}\sum_{k=0}^{K}\sqrt{|p_k|}\,
  e^{ik\phi}|k\rangle,
  \qquad
  \mathcal{N}_K=\sum_{k=0}^{K}|p_k|,
  \label{eq:RK}
\end{equation}
is an arbitrary superposition of the lowest $K+1$ Fock states.
 Here the right state is normalized by itself $\langle R_K | R_K \rangle =1$.
 This state can be synthesised by alternating displacement and number-selective phase
(SNAP) layers, or by sideband pulses in the ion
case~\cite{Hofheinz2009, Heeres2015, Krastanov2015, Heeres2017, Leibfried2003}. Basis up to $K\simeq10$--$20$ are established technology.

\emph{(ii) The metric.} Although $\eta$ was defined in Eq.~\eqref{eq:eta}
as an infinite diagonal sum, the sign pattern collapses it to
\begin{equation}
  \eta=2|0\rangle\langle0|-(-1)^{\hat N}
  \label{eq:etaparity}
\end{equation}
for every $p$ in the domain of existence. The metric is therefore
\emph{the parity together with a single number-selective phase on the
vacuum}---not a generic SNAP gate. Controlled-parity is among the
best-developed primitives in circuit QED, where it underlies both Wigner tomography and
quantum non-demolition photon counting~\cite{Bertet2002, Vlastakis2013, Sun2014}. In the
optical domain the same statistics are accessible with photon-number-resolving
transition-edge sensors~\cite{Lita2008, Eisaman2011}. Note also that
$\eta$ is manifestly $K$-independent, so the truncated construction uses
the same metric as the exact one.

\emph{(iii) Joint readout.} 
We read out
 the particle numbers $n_A$ and $n_B$ with the signs $s_{n_A}$ and $s_{n_B}$
for a combined state $\ket{R_{A,K}R_{B,K}} $. The total particle number is $N=n_A+n_B$ of the microscopic pair $(k,N-k)$, with $n_A=k$ and $n_B=N-k$.
The corresponding sign-weighted expectation value is
\begin{align}
\bra{R_{A,K}R_{B,K}}\,\eta_A\otimes\eta_B\,\Pi_N\,\ket{R_{A,K}R_{B,K}}
= \frac{Q_K(N)}{\mathcal{N}_K^{2}} ,
\label{eq:jointreadout}
\end{align}
where $\Pi_N=\sum_{j+m=N}\ket{j,m}\bra{j,m}$ and the factor
$\mathcal{N}_K^{-2}$ arises because each mode of the product state carries the
normalization $\mathcal{N}_K^{-1/2}$ of Eq.~\eqref{eq:RK}.
Operationally the right-hand
side is a \emph{signed} average: one records the pair $(j,m)$ and accumulates the sign
$s_js_m$ for every shot with $j+m=N$, rather than a bare count.

In addition, we can also evaluate the individual \emph{signed} average, which is
\begin{align}
 \bra{R_{A,K}}\,\eta_A\,\sum_{j}\ket{j}\bra{j}\ket{R_{A,K}}
= \frac{F_K}{\mathcal{N}_K}.
\label{eq:indreadout}
\end{align}
This quantity fixes the normalization used to form
$Q(N)=\lim_{K\to\infty}Q_K(N)/F_K^2$.

\subsection{What truncation does and does not affect}

First, we write $F_K\equiv\sum_{k\le K}p_k$ for the signed partial sum and let $Q_K(N)=\sum_{j+m=N,\;j,m\le K}p_jp_m$.
We have $\sum_{N=0}^{2K}{Q_K(N)} = F_K^2$.

\begin{proposition}[Truncation exactness]
\label{prop:trunc}
For every $N\le K$ one has $Q_K(N)=Q(N)$ exactly.
\end{proposition}

\begin{proof}
Any pair with $j+m=N\le K$ satisfies $j\le N\le K$ and $m\le N\le K$
automatically, so the constraint $j,m\le K$ is inactive and the sum
coincides term by term with the untruncated convolution of
Eq.~\eqref{eq:QNdef}. Truncation can affect only the manifolds
$K<N\le2K$.
\end{proof}

Two corollaries follow immediately and organise the rest of the section.

\begin{corollary}[The fused qubit is exact]
\label{cor:exactqubit}
The $N\le1$ component of $|R_K\rangle_A\otimes|R_K\rangle_B$ involves
only $p_0$ and $p_1$ and is therefore independent of $K$. Explicitly,
$Q_K(0)=p_0^2=q$ and $Q_K(1)=2p_0p_1=p$ for every $K\ge1$, so the state
conditioned on $N\le1$ is
\begin{equation}
  \sqrt{q}\,|0_{\rm log}\rangle+\sqrt{p}\,e^{i\phi}|1_{\rm log}\rangle,
  \label{eq:heraldedstate}
\end{equation}
\emph{exactly}, at any cutoff. Truncation affects only the normalization $\mathcal{N}_K$.
\end{corollary}

\subsection{$K=3$ as an example at the unbiased point}
\label{sec:K3}
\begin{table*}[t]
\caption{Complete outcome statistics for two truncated half-qubits at
$K=3$, $p=q=1/2$. All weights are exact integers in units of $1/729$.
Within the manifolds $N=2$ and $N=3$ the positive and negative sign
products carry \emph{identical} weight, so the signed column vanishes
identically: this is the Vandermonde cancellation of
Eq.~\eqref{eq:QNdef}, visible directly in the sign-weighted photon-number statistics. The residual
entries at $N=4,5,6$ are truncation debris and lie, as guaranteed by
Proposition~\ref{prop:trunc}, entirely above the cutoff. The second-last column is the
sign-weighted photon-number joint-measurement outcome $Q_3(N)/\mathcal{N}_K^2$ of Eq.~\eqref{eq:jointreadout}
(the signed count divided by $\mathcal{N}_3^2=729/512$, i.e.\ by $729$ in these units);
the last column rescales it by $\mathcal{N}_K^2/F_K^2$ to give the normalized distribution
$Q_3(N)/F_K^2$, which converges to $Q(N)$ as $K\to\infty$.}
\label{tab:K3}
\begin{ruledtabular}
\begin{tabular}{crrrrrr}
$N$ & partitions & $s_js_m=+1$ & $s_js_m=-1$ & signed &
$Q_3(N)/\mathcal{N}^2_K$ & $Q_3(N)/F^2_K$ \\
\colrule
$0$ & $(0,0)$                     & $256$ & $0$  & $256$ & $256/729$ & $256/529$ \\
$1$ & $(0,1),(1,0)$               & $256$ & $0$  & $256$ & $256/729$ & $256/529$ \\
$2$ & $(1,1)\,;\,(0,2),(2,0)$     & $64$  & $64$ & $0$   & $0$ & $0$ \\
$3$ & $(0,3),(3,0)\,;\,(1,2),(2,1)$ & $32$ & $32$ & $0$  & $0$ & $0$ \\
\colrule
$4$ & $(1,3),(3,1),(2,2)$         & $20$  & $0$  & $20$  & $20/729$ & $20/529$ \\
$5$ & $(2,3),(3,2)$               & $0$   & $4$  & $-4$  & $-4/729$ & $-4/529$ \\
$6$ & $(3,3)$                     & $1$   & $0$  & $1$   & $1/729$ & $1/529$ \\
\colrule
    & total                       & $629$ & $100$& $529$ & $529/729$ & $1$ \\
\end{tabular}
\end{ruledtabular}
\end{table*}

Everything above is illustrated most transparently for the $K=3$ case. Take $p=q=1/2$ and truncate at $K=3$. The four
retained Sz\'ekely coefficients are
\begin{equation}
  p_0=\tfrac{1}{\sqrt2},\quad
  p_1=\tfrac{1}{2\sqrt2},\quad
  p_2=-\tfrac{1}{8\sqrt2},\quad
  p_3=\tfrac{1}{16\sqrt2},
\end{equation}
with signs $s=(+,+,-,+)$, so that
\begin{equation}
  \mathcal{N}_3=\sum_{k\le3}|p_k|=\frac{27}{16\sqrt2},
  \qquad
  F_3=\sum_{k\le3}p_k=\frac{23}{16\sqrt2}.
\end{equation}
The truncated right state~\eqref{eq:RK} is then remarkably simple: the
Born weights $|p_k|/\mathcal{N}_3$ are the rationals
$\tfrac{16}{27},\tfrac{8}{27},\tfrac{2}{27},\tfrac{1}{27}$, and
\begin{equation}
  |R_3\rangle=\frac{1}{3\sqrt3}
  \Bigl(4|0\rangle+2\sqrt2\,e^{i\phi}|1\rangle
  +\sqrt2\,e^{2i\phi}|2\rangle+e^{3i\phi}|3\rangle\Bigr).
  \label{eq:R3}
\end{equation}
This is a four-rung state with a single photon-number distribution
$(16,8,2,1)/27$---well within reach of existing cavity or motional-state
 tomography.

Preparing two such modes and recording the pair $(j,m)$ gives sixteen
outcomes with joint weights $w_jw_m$, which in units of $1/729$ are the
products of $(16,8,2,1)$. Grouping them by the total excitation
$N=j+m$ and by the sign product $s_js_m$ yields
Table~\ref{tab:K3}.

Three features of Table~\ref{tab:K3} deserve emphasis.

\emph{(i) The theoretical sign-weighted cancellation is exact.} At
$N=2$ the partition $(1,1)$ carries weight $64$ with $s_1s_1=+1$, while
$(0,2)$ and $(2,0)$ carry weight $32$ each with $s_0s_2=-1$; the two
balance exactly. The same happens at $N=3$, where
$(0,3),(3,0)$ balance $(1,2),(2,1)$ at $32$ apiece. 
 In a finite data set, these reconstructed positive and negative weights fluctuate
statistically; exact equality is a statement about the theoretical distribution, not
about every experimental sample.

\emph{(ii) The logical populations are exact.} Reading off
$Q_3(1)/ Q_3(0)=1=p/q$ recovers the unbiased qubit
with no truncation error whatsoever, in accordance with
Corollary~\ref{cor:exactqubit}. Only the common scale is affected:
$ Q_3(0)/F_K^2= Q_3(1)/F_K^2=256/529=0.4839$ rather than $1/2$,
a $3.2\%$ deficit equal to $F_3^{-2}-1=-17/529$. Note that $F_3>1$ here:
the signed partial sums alternate about their limit, so this residual
changes sign with $K$.

\emph{(iii) The nonvanishing terms are above the cutoff.} The manifolds $N=4,5,6$
carry total signed weight $17/529$ and are pure truncation artefact;
they would be absent for the untruncated state.

We can consider larger $K$ to see convergence toward the standard qubit.
For $K=5,10,20,50$, we can evaluate $Q_K(N=0,1)/F_K^2=0.492, 0.503,0.501,0.500$, which can recover the standard qubit with larger finite truncation.

\section{Generalization to $n$th roots: the $1/n$th qubit}
\label{sec:nthroot}

Nothing in Sec.~\ref{sec:halfqubit} requires the root to be square. Taking
$f(z)=(q+pz)^{1/n}$ gives generalised quasi-probabilities
$p_k^{(n)}=q^{1/n}\binom{1/n}{k}(p/q)^k$ whose signs, for $0<1/n<1$,
follow the \emph{same} alternating pattern~\eqref{eq:signs}; the
metric~\eqref{eq:etaparity} is therefore independent of $n$. The
$n$-fold Vandermonde identity
$\sum\binom{1/n}{k_1}\cdots\binom{1/n}{k_n}=\binom{1}{N}$ annihilates
every manifold $N\ge2$ exactly as before, so $n$ such objects fuse into
one qubit: a \emph{$1/n$th qubit}. The $L_1$ norm generalizes to
\begin{equation}
  \|p\|_1^{(n)}=2q^{1/n}-(q-p)^{1/n},
  \qquad
  \sup_p\|p\|_1^{(n)}=2^{(n-1)/n},
  \label{eq:nthroot}
\end{equation}
reducing to Eq.~\eqref{eq:L1} at $n=2$, with
$\langle Z_A\rangle_{bi}=(q-p)^{1/n}$ the local $n$th root of the macroscopic
expectation value. The supremum $2^{(n-1)/n}$ increases monotonically $n$ and
approaches $2$, the full Hilbert-space dimension of the qubit, as $n\to\infty$: the finer
the fractionalization, the closer its $L_1$ norm comes to saturating the dimension of
the object being fractionalized. The tail becomes heavier,
$|p_k^{(n)}|\sim k^{-1-1/n}$, so the series remains absolutely
convergent for every finite $n$ but the truncation cost grows.

\section{Limited comparison with Majorana fractionalization}
\label{sec:majorana}

Majorana fermions provide a useful contrast rather than a microscopic interpretation
of the present construction. In an Ising-type topological phase, defects binding Majorana
fermions realize the non-Abelian charge $\sigma$, whose quantum dimension is
$d_\sigma=\sqrt2$; the fermionic charge $\psi$ has $d_\psi=1$, and an individual
Majorana operator is not itself an anyon carrying quantum dimension
$\sqrt2$~\cite{Read2000,Nayak2008,Alicea2012,Sarma2015}. Two  define
one nonlocal complex-fermion mode, whereas a parity-preserving logical qubit is
conventionally encoded using four Majorana fermions at fixed total parity. By contrast, the
half-qubit here is a state-specific signed-convolution factor in a Krein representation.
Its $L_1$ norm is representation dependent and is neither a fusion-category quantum
dimension nor evidence for braiding, topological degeneracy, or topological protection.
Consequently,
$\sup_p\lVert\boldsymbol p\rVert_1=\sqrt2=d_\sigma$ is only a numerical analogy.
Likewise, the state-specific $n$th-root factor (the $1/n$-qubit) has no established
correspondence with a $\mathbb{Z}_n$ parafermion.

\section{Conclusion}
\label{sec:conclusion}

We have constructed a local, algebraic fractionalization of a qubit by embedding
Sz\'{e}kely's half-coin into a Krein space. Two biorthogonal half-qubits fuse into an
arbitrary pure qubit such that the collective
$N\ge2$ vectors are null in Krein space under the signed Vandermonde
convolution. As a consequence, the $L_1$ norm
$\lVert p\rVert_1 = 2\sqrt q-\sqrt{q-p}$ rises monotonically with the bias and
attains $\sqrt2$ exactly at the point $p=q=1/2$.
This algebraic fractionalization differs from Majorana fractionalization.
A  single half-qubit has a well-defined state space with negative probability interpretation, which can be incorporated in the property of non-Hermiticity.

Furthermore, we find that the local algebra is a
$\mathbb{Z}_2\times\mathbb{Z}_2$ group $\{I,Z_A,\eta,V\}$ generated by the parity
and the $\eta$-metric.  No $\eta$-self-adjoint operators can anticommute with
number parity. On the selected positive fused qubit, however, produces
encoded $X_{\rm L},Y_{\rm L},Z_{\rm L}$ satisfying the Pauli algebra.

The scheme is not merely formal. Because the $N\le1$ manifolds involve only $p_0$ and
$p_1$, the fused qubit is exact at every Fock cutoff, and only quantities built from the
$L_1$ norm inherit the slow $k^{-3/2}$ tail. For any truncation of the Fock states $|K\rangle$ with $N \le K$ and $N$ being the total number , the Vandermonde cancellation is  exactly as established by {\bf Proposition 2.}
The required state is a 
superposition with finite photon-number weights, which can be reached within the existing cavity
and trapped-ion synthesis. The readout uses variants of parity and signed photon number measurements are already standard in circuit QED, providing several experimental accessible  platform for realization of a half-qubit.

Several questions remain open. 
First, is it possible to create an entangled two-half-qubit? If we can realize this entangled pair, we might use it as a quantum resource and can apply the quantum teleportation scheme.
 A second question is whether the physical subspace can be characterized algebraically rather than
postulated---that is, whether there exists a class of operators leaving the $N\le1$ sector
invariant, which would place the present truncation on the same footing as the
Gupta--Bleuler quantization of the electromagnetic field~\cite{Gupta1950, Bleuler1950},
where the null states form a subspace preserved by every gauge-invariant observable.
A third question is about the dynamics of a half-qubit. The construction here is entirely kinematical, 
and we make no claim of an operational advantage. Whether there exists a Hamiltonian acting on half-qubits whose
fused evolution reproduces a prescribed qubit channel more economically than the qubit
description itself would determine whether fractional quasi-probability degrees of freedom
are a useful computational resource or only a structural curiosity. Finally, the $n$th-root
generalization of Sec.~\ref{sec:nthroot} raises the question of whether the sequence of
$1/n$-qubits, with the $L_1$ norm $2^{(n-1)/n}$ approaching the Hilbert-space dimension
$2$, admits a continuum limit.

\begin{acknowledgments}
The author used generative-AI tools (ChatGPT Sol 5.6, Claude Opus 5, and Gemini 3.1 Pro)
for language assistance and algebraic cross-checks. The author reviewed and verified the
derivations, numerical values, citations, and final manuscript text and assumes full
responsibility for the work.
The  author thank RIKEN iTHEMS and the NCTS, where this work has been initiated from during the workshop on ''The 2nd iTHEMS-NCTS Joint Workshop".
This work is supported by National Science and Technology Council of Taiwan under Grants No. NSTC 114-2918-I-007-015, 115-2112-M-007-039. 
The author thanks the support from the National Center for Theoretical Sciences, Physics Division. 
\end{acknowledgments}

\appendix
\section{Collective Radical-Quotient Theorem}
\label{APP}

Let
\[
\eta_{AB}\equiv\eta_A\otimes\eta_B
\]
and define the restricted Krein form by
\begin{equation}
 [u,w]_{\eta_{AB}}
 \equiv \bra{u}\eta_{AB}\ket{w}.
 \label{eq:restricted-krein-form}
\end{equation}
The collective subspace and its higher-sector subspace are
\begin{align}
 \mathcal S
 &\equiv
 \overline{\operatorname{span}}
 \left\{\ket{\bar N}_R:N\geq0\right\},\\
 \mathcal K
 &\equiv
 \overline{\operatorname{span}}
 \left\{\ket{\bar N}_R:N\geq2\right\},
 \label{eq:collective-subspaces}
\end{align}
where both closures are taken in the auxiliary Hilbert-space norm.
The radical of the restricted form is
\begin{align}
 \operatorname{rad}\mathcal S
 &\equiv
 \left\{
 \ket{u}\in\mathcal S:
 [u,w]_{\eta_{AB}}=0
 \text{ for every }\ket{w}\in\mathcal S
 \right\}  \notag\\
& =
 \mathcal S\cap\mathcal S^{[\perp]} .
 \label{eq:radical-definition}
\end{align}

\begin{theorem}[Collective radical-quotient theorem]
\label{thm:radical-quotient}
Let $p+q=1$ and $0<p\leq q$. Then
\begin{equation}
 \mathcal K=\operatorname{rad}\mathcal S.
 \label{eq:radical-equality}
\end{equation}
Consequently, the Krein form induces a nondegenerate,
positive-definite form on $\mathcal S/\mathcal K$. The quotient is
two dimensional, and its Gram matrix in the classes
$[\ket{\bar0}_R]$ and $[\ket{\bar1}_R]$ is
\begin{equation}
 G_{\mathcal S/\mathcal K}
 =
 \begin{pmatrix}
 q&0\\
 0&p
 \end{pmatrix}.
 \label{eq:quotient-gram}
\end{equation}
\end{theorem}

\begin{proof}
The collective vectors $\ket{\bar N}_R$ belonging to different
total-number sectors are orthogonal in the auxiliary Hilbert inner
product. For $p>0$, every $\ket{\bar N}_R$ is nonzero, with
\begin{equation}
 r_N\equiv
 \left\|\ket{\bar N}_R\right\|_{\mathcal H}^{2}
 =
 \sum_{k=0}^{N}|p_k|\,|p_{N-k}|>0.
 \label{eq:collective-hilbert-norm}
\end{equation}
Consequently, every vector $\ket{u}\in\mathcal S$ has a unique
Hilbert-convergent expansion
\begin{equation}
 \ket{u}=\sum_{N=0}^{\infty}a_N\ket{\bar N}_R,
 \qquad
 \sum_{N=0}^{\infty}|a_N|^2r_N<\infty,
 \label{eq:u-expansion}
\end{equation}
and similarly
\begin{equation}
 \ket{w}=\sum_{N=0}^{\infty}b_N\ket{\bar N}_R.
\end{equation}

Because $\eta_{AB}$ is bounded, the Krein form is continuous in the
auxiliary Hilbert topology. Therefore, the sector-overlap identity
\begin{equation}
 [\bar N,\bar M]_{\eta_{AB}}
 =
 {}_L\!\braket{\bar N|\bar M}_R
 =
 \delta_{NM}Q(N)
 \label{eq:appendix-sector-overlap}
\end{equation}
extends from finite linear combinations to all vectors in
$\mathcal S$. Vandermonde's identity gives
\begin{equation}
 Q(N)
 =
 q\left(\frac{p}{q}\right)^N\binom{1}{N},
\end{equation}
and hence
\begin{equation}
 Q(0)=q,\qquad Q(1)=p,\qquad
 Q(N)=0\quad(N\geq2).
 \label{eq:appendix-Q-values}
\end{equation}
It follows that
\begin{equation}
 [u,w]_{\eta_{AB}}
 =
 q\,a_0^*b_0+p\,a_1^*b_1.
 \label{eq:restricted-form-explicit}
\end{equation}

If $\ket{u}\in\mathcal K$, then $a_0=a_1=0$, and
Eq.~\eqref{eq:restricted-form-explicit} gives
\[
 [u,w]_{\eta_{AB}}=0
\]
for every $\ket{w}\in\mathcal S$. Thus
\[
 \mathcal K\subseteq\operatorname{rad}\mathcal S.
\]

Conversely, suppose that
$\ket{u}\in\operatorname{rad}\mathcal S$. Choosing successively
$\ket{w}=\ket{\bar0}_R$ and
$\ket{w}=\ket{\bar1}_R$ gives
\begin{equation}
 q\,a_0^*=0,\qquad p\,a_1^*=0.
\end{equation}
Since $p,q>0$, this implies $a_0=a_1=0$, and therefore
$\ket{u}\in\mathcal K$. Hence
\[
 \operatorname{rad}\mathcal S\subseteq\mathcal K,
\]
which proves Eq.~\eqref{eq:radical-equality}.

Because $\mathcal K$ is the radical, the quotient form
\begin{equation}
 \bigl[[u],[w]\bigr]_{\mathrm{quot}}
 \equiv [u,w]_{\eta_{AB}}
 \label{eq:quotient-form}
\end{equation}
is independent of the chosen representatives. Every quotient class
has a representative
\[
 a_0\ket{\bar0}_R+a_1\ket{\bar1}_R,
\]
and Eq.~\eqref{eq:restricted-form-explicit} gives
\begin{equation}
 \bigl[[u],[u]\bigr]_{\mathrm{quot}}
 =
 q|a_0|^2+p|a_1|^2.
\end{equation}
This expression is strictly positive for every nonzero quotient
class because $p,q>0$. Thus the induced form is positive definite.
The quotient is spanned by the two independent classes
$[\ket{\bar0}_R]$ and $[\ket{\bar1}_R]$, and its Gram matrix is
Eq.~\eqref{eq:quotient-gram}.
\end{proof}

\paragraph*{Degenerate endpoint $p=0$.}
At $p=0$ one has $q=1$, $p_0=1$, and $p_k=0$ for every $k\geq1$.
Consequently,
\[
 \ket{\bar0}_R=\ket{0,0},
 \qquad
 \ket{\bar N}_R=0\quad(N\geq1).
\]
Therefore,
\[
 \mathcal S=\operatorname{span}\{\ket{0,0}\},
 \qquad
 \mathcal K=\operatorname{rad}\mathcal S=\{0\},
\]
and $\mathcal S/\mathcal K$ is the one-dimensional vacuum space.
In particular, no normalized $N=1$ quotient representative exists
at this endpoint.

\bibliography{half_qubit.bib}
\end{document}